\documentclass[12pt,doublecolumn]{IEEEtran}
\usepackage{amsmath,amssymb,latexsym,float,epsfig,epstopdf,subfigure,graphicx,array,algorithm,algpseudocode,booktabs}
 \usepackage{bm}
\usepackage{xcolor}
 \usepackage[noadjust]{cite}
\usepackage[utf8]{inputenc}
\usepackage[T1]{fontenc}
 \usepackage{multirow}

\usepackage{amsmath,amsthm}
\theoremstyle{plain}
\usepackage{algorithm} 
\usepackage{algpseudocode} 
 \usepackage{flushend}
    \usepackage{lipsum}

\newtheorem{thm}{\textbf{Theorem}}

\usepackage{xpatch}
\makeatletter
\xpatchcmd{\@thm}{\thm@headpunct{.}}{\thm@headpunct{}}{}{}
\makeatother
\makeatletter

\renewenvironment{proof}[1][\proofname]{\par
\pushQED{\qed}%
\normalfont 
\trivlist
\item\relax
{\itshape
#1\@addpunct{:}}\hspace\labelsep\ignorespaces
}{%
\popQED\endtrivlist\@endpefalse
}
\begin{document}

\title{Random Fourier Feature Based Deep Learning for Wireless Communications}
\author{Rangeet~Mitra and Georges~Kaddoum
\thanks{Rangeet Mitra and Georges Kaddoum are with the Resilient Machine learning Institute (ReMI) of the  \'{E}cole de Technologie Sup\'{e}rieure (\'{E}TS), University of Quebec, Montreal, Canada. (Email: rangeet.mitra.1@ens.etsmtl.ca). 

This paper is under review in the IEEE Transactions on Vehicular Technology (Submitted on September 17, 2020). A version of this paper was under was previously submitted to the IEEE Transactions on Neural Networks and Learning Systems on March 20, 2020; however, this was submitted to the IEEE Transactions on Vehicular Technology after its decision on Sept. 1, 2020, which advised rejection of this work due to its better suitability to a Communications journal.}
}
\maketitle
\begin{abstract}
Deep-learning (DL) has emerged as a powerful machine-learning technique for several classic problems encountered in generic wireless communications. Specifically, random Fourier Features (RFF) based deep-learning has emerged as an attractive solution for several machine-learning problems; yet there is a lacuna of rigorous results to justify the viability of RFF based DL-algorithms in general. To address this gap, we attempt to analytically quantify the viability of RFF based DL. Precisely, in this paper, analytical proofs are presented demonstrating that RFF based DL architectures have lower approximation-error and probability of misclassification as compared to classical DL architectures. In addition, a new distribution-dependent RFF is proposed to facilitate DL architectures with low training-complexity. Through computer simulations, the practical application of the presented analytical results and the proposed distribution-dependent RFF, are depicted for various machine-learning problems encountered in next-generation communication systems such as: a) line of sight (LOS)/non-line of sight (NLOS) classification, and b) message-passing based detection of low-density parity check codes (LDPC) codes over nonlinear visible light communication (VLC) channels. Especially in the low training-data regime, the presented simulations show that significant performance gains are achieved when utilizing RFF maps of observations. Lastly, in all the presented simulations, it is observed that the proposed distribution-dependent RFFs significantly outperform RFFs, which make them useful for potential machine-learning/DL based applications in the context of next-generation communication systems.
\end{abstract}
\section{Introduction}
The capacity of classical machine learning methodologies are limited in terms of learning accurate representations from data and in generalizing over large datasets \cite{bengio1,bengio2009learning}. On the other hand, deep-learning (DL) has emerged as a viable machine-learning paradigm to model a nonlinear/abstract mapping from observations or learning representations from data. Furthermore, DL based algorithms have been successfully deployed in numerous sub-domains like computer-vision, speech processing, natural language processing, wireless communications, and time-series prediction. For various tasks/problems in these sub-domains, several DL-architectures, e.g. multilayer perceptron, convolutional neural network (CNN), and recurrent neural network (RNN) \cite{li2015fpga,krizhevsky2012imagenet,kim2014convolutional} that are optimized using the backpropagation algorithm are proposed. Further, long-short term memory (LSTM) based DL architectures are found to be particuarly viable, as LSTMs address the issue of exploding/vanishing gradients \cite{bengio1994learning,hochreiter1997long} encountered in the backpropagation algorithm for modelling/predicting dynamical systems with memory. However, in spite of DL enjoying widespread deployment for complex machine-learning tasks, deep neural networks (DNN) have been found to be sensitive to hyperparameters like number of layers, number of hidden nodes in each layer and the nature of the activation functions.

On the other hand, classical kernel based learning techniques are well-known for their ability to model high-dimensional representations and for their generalization \cite{scholkopf2001learning,liu2011kernel,chen2013system}, and have fewer hyperparameters that require optimization as compared to DNNs; however they require the representation of the learning-parameter to be expressed as an implicit inner-product in a reproducing kernel Hilbert space (RKHS) using Mercer kernels. The exact nature of the implicit feature-map is unknown; however the feature-map can be well-approximated explicitly using sampling methods like random Fourier features (RFF)
 \cite{8835057,8830377}, which can be further utilized as features for potential DL applications. 
From simulations presented in various works in the literature it is observed that the use of RFF in DNNs, significantly boosts the performance of the DNNs rather than utilizing the indigenous features \cite{mehrkanoon2018deep}. Moreover, the RFFs are approximations of feature maps, which facilitate intrinsically regularized parameter updates. This in turn leads to improved generalization \cite{liu2011kernel,theodoridis2015machine} in RFF based DNN architectures, which has prompted the proposal of several DL based architectures which attempt to highlight the viability of RFF-maps through extensive simulation studies \cite{mehrkanoon2018deep,mehrkanoon2019cross}. 
However, existing works on RFF based DL motivate their results on intuition/simulation examples rather than providing rigorous analytical results to establish the paradigm of RFF based DL. Furthermore, RFFs in general require a large number of dimensions to gain an accurate approximation of an RKHS, which significantly increases the overall computational complexity, and creates a serious implementation bottleneck in the practical deployment of  RFF based DL. Hence, based on this review, we highlight the following novelty points of this work:
\begin{itemize}
\item We seek to quantify the viability of RFF-maps in the context of DNN rigorously. We present our claims in the form of two theorems and provide detailed proofs to justify the benefits of utilizing RFF-mapped observations for DL. 
\item To overcome the computational complexity incurred by RFF-mappings, a distribution-dependent RFF is proposed which outperforms classical RFF in scenarios with low/medium RFF dimensions, and delivers better classification performance with lesser amount of training data and with lesser amount of training-data.  
\end{itemize}
The paradigm of RFF based DNN is tested over the following practical problems encountered in the context of next-generation communication systems: a) line of sight (LOS)/non-line of sight (NLOS) identification for wireless links using LSTM based DNN, and b) Message-passing based low-density parity check (LDPC) decoding over nonlinear visible light communication (VLC) channels. Next, we present existing works on the two aforementioned sub-domains.
\subsection{LOS/NLOS classification for wireless links}
Accurate inference of channel-state is critical for node-localization, and link-adaptation over ad-hoc tactical networks, which necessitates extracting accurate information of the channel-conditions and tracking the users' channels. However, in high mobility scenarios, inferring accurate information about the channel state is quite challenging, mainly due to the time-varying nature of the wireless channel, which significantly impairs localization/degrades the overall wireless link due to detrimental outages caused by NLOS scenarios \cite{mao2018probabilistic,joo2019deep,guvenc2007nlos}. Hence, it is quite essential to develop accurate signal processing algorithms which estimate and track the channel, and also infer the channel type, i.e. LOS vs NLOS such that suitable link-adaptation or network topology selection can be performed \cite{guvencc2007nlos}. 

In this section, we focus on reviewing signal processing algorithms in general for LOS/NLOS detection or localization. Several LOS/NLOS detection methodologies are proposed using popular temporal DL paradigms trained on the receive signal strength indicator (RSSI) \cite{choi2017deep,nguyen2018nlos} based on LSTM/hybrid CNN \cite{goodfellow2016deep}. Apart from this, an LSTM trained with local temporal RSSI features is proposed for indoor localization in \cite{chen2019wifi}. Furthermore, there are unsupervised approaches for LOS/NLOS identification including the use of a Gaussian mixture model for LOS/NLOS classification \cite{fan2019non}. Moreover, the work in \cite{joo2019deep} suggests tracking V2V channels using IEEE 802.11p, and it is particularly highlighted that the NLOS components cause link-outages due to packet losses, and that accurate LOS/NLOS detection is needed to mitigate such losses a-priori by link-adaptation. Furthermore, it is noteworthy that the task of LOS/NLOS detection is more difficult for outdoor scenarios which are characterized by longer ranges, and higher mobility and delay-spread. 

From the above review, it can be concluded that fast and accurate LOS/NLOS identification is essential, and LSTMs have been found useful for LOS/NLOS prediction in general. However, in outdoor scenarios which are characterized by high mobility (i.e. with typically lower coherence time as compared to indoor scenarios), it is quite essential to perform accurate LOS/NLOS predictions with lesser training data. Hence, we outline the following novel points of our work which seek to enhance the accuracy of LOS/NLOS classification using hybrid RFF/LSTM based DNNs in the low training-data regime:
\begin{itemize}
\item This work presents an analytical result that justifies the viability of mapping of the incoming observations to RFF for training neural-network architectures like LSTMs. The gains promised by the presented analytical results are validated using computer simulations for  LOS/NLOS identification of outdoor wireless channels.
\item A new kind of distribution-dependent RFF is proposed which is found to deliver-improved approximation of RKHS  with less number of RFF-dimensions, and hence facilitates low-complexity architectures. When used in conjunction with LSTM, the distribution-dependent RFF is found to achieve better $F_{1}$-score for LOS/NLOS classification which lowers the overall computational complexity for a given error-floor. 
\end{itemize}
\subsection{LDPC-decoding for static VLC channels}
Low-density parity check (LPDC) codes, well-known as one of the capacity-achieving codes, have been widely deployed for both radio-frequency (RF) and VLC based communication systems \cite{ryan2009channel,tang2015analysis}. However, in the context of VLC, the algebraic structure of the codewords is distorted due to nonlinear LED transfer-characteristics, which, if unmitigated, severely impairs message-passing based detection. In this work, the observed codewords from a nonlinear VLC channel are iteratively mapped to an approximate RKHS using RFF to mitigate the LED nonlinearity and recover the codewords, where message-priors are iteratively updated based on the ``quality" of RFF-approximation at each iteration. From the conducted simulations over nonlinear VLC channels, a significant BER gain is observed when performing message-passing based detection using the RFF-map. Furthermore, one can observe a significant BER-performance gain upon deployment of our proposed distribution-dependent RFFs as compared to RFFs, which renders the proposed RFFs viable.
\section{Overview of RFF based signal processing}
In this section, we provide an overview of RFF based signal processing. Using an implicit feature map to RKHS (denoted by $\bm{\Phi}:\mathbb{R}^{n}\to\mathcal{H}$), and invoking the Representer theorem \cite{scholkopf2001generalized}, an arbitary function $f(\cdot)$ may be represented as the following weighted combination:
\begin{gather}
f(\cdot) = \sum\limits_{\forall j}\beta_{j}\kappa(\mathbf{x}_{j},\cdot),
\end{gather}
where $\mathbf{x}_{j}$ denotes the $j^{th}$ observation, $\beta_{j}$ denote the approximation-weights, and $\kappa(\cdot,\cdot):\mathbb{R}^{n}\times\mathbb{R}^{n}\to\mathbb{R}$ is a continuous and shift-invariant Mercer kernel. Estimating the above representation of $f(\cdot)$ is computationally involved and requires expressing $f(\cdot)$ \textit{only} in terms of Mercer kernels which prevents us from gaining intuitive insights (as opposed to insights provided by intermediate layers of DNN). To reap the benefits of RKHS based approaches (like regularization, generalization etc) by potentially deploying them as features in DNN, the Mercer kernel $\kappa(\cdot,\cdot)$ can be well-approximated as an RFF \cite{bouboulis2017online}. This approximation is motivated by the Bochner's theorem \cite{bochner2005harmonic}, which is restated as below: 
\begin{thm}
A continuous and shift-invariant kernel $\kappa(\cdot,\cdot):\mathbb{R}^{n}\times\mathbb{R}^{n}\to\mathbb{R}$ is positive-definite iff it is the Fourier transform of a Borel measure $\rho(\cdot)$ on $\mathbb{R}^{n}$. 
\end{thm}
Using Bochner's theorem, a positive-definite kernel can be expressed as
\begin{gather}
\kappa(\bm{a},\bm{b}) = \int\limits_{\bm{\omega}}\exp(j\bm{\omega}^{T}[\bm{a}-\bm{b}])\rho(\bm{\omega})d\bm{\omega}.
\end{gather}
Denoting $\hat{\bm{\Phi}}_{\bm{\omega}}(\mathbf{a}) = \exp(j\bm{\omega}^{T}\mathbf{a})$, one can re-write $\kappa(\bm{a},\bm{b})$ (where $\hat{\bm{\Phi}}_{\bm{\omega}}:\mathbb{R}^{n}\to \mathbb{R}^{n_{G}}$) as
\begin{gather}
\kappa(\bm{a},\bm{b})=\mathbb{E}_{\bm{\omega}\sim\rho(\bm{\omega})}[\hat{\bm{\Phi}}_{\bm{\omega}}(\bm{a})\hat{\bm{\Phi}}_{\bm{\omega}}(\bm{b})].
\end{gather}
Lastly, to lower the approximation-error, the above mean may be approximated by an sample average such as
\begin{gather}
\kappa(\bm{a}-\bm{b})\approx\frac{1}{n_{G}}\sum\limits_{i=1}^{n_{G}}\hat{\bm{\Phi}}_{\bm{\omega}}(\bm{a})^{T}\hat{\bm{\Phi}}_{\bm{\omega}}(\bm{b}).
\end{gather}
Further bounds on the error in kernel-approximation were derived in \cite{rahimi2008random} using the RFF based approximation of feature-maps. In particular, for a real Gaussian kernel, an RFF (denoted here by $\hat{\bm{\Phi}}:\mathbb{R}^{n}\to\mathbb{R}^{n_{G}}$) is obtained as
\begin{gather}
\label{RFF}
\hat{\bm{\Phi}}(\textbf{x}) = \sqrt{\frac{2}{n_{G}}}\begin{bmatrix} 
    \cos(\bm{\omega}_{1}^{T}\mathbf{x}+b_{1}) \\
    \vdots  \\
    \cos(\bm{\omega}_{n_{G}}^{T}\mathbf{x}+b_{n_{G}})  
    \end{bmatrix},
\end{gather}
where each $\{\bm{\omega}_{i}\}_{i=1}^{n_{G}}$ is a Gaussian vector, with zero mean and covariance $\frac{1}{\sigma^2}\mathbb{I}_{n_{G}}$, with $\mathbb{I}_{n_{G}}$ denoting the identity matrix of size $n_{G}$, {$\sigma$ denotes the kernel-width, and $(\cdot)^{T}$ denotes the transpose operation.

It is noted that since an RKHS is closed, an exact representation exists for a wide class of functions in an RKHS \cite{scholkopf2001generalized}, which makes RKHS based learning methods suitable for function-approximation. However, most RKHS techniques rely on learning a dictionary of observations \cite{bouboulis2011extension,chen2012quantized,liu2011kernel}, and hence are sensitive to inclusion of erroneous entries due to noise in the initial learning-stages, and also makes practical implementation complex. In this regard, RFFs provide an approximate explicit map to RKHS, which facilitates practical implementations with a finite memory budget whilst delivering equivalent performance as promised by RKHS based learning algorithms, which make them promising for RFF based DNN-architectures.    
\section{Proof of viability of RFF for training LSTM}
In this section, a proof is outlined that guarantees the viability of RFF based LSTMs in terms of classification accuracy/optimizing the hit-or-miss cost function. 
First of all, we enlist the considered assumptions as follows:
\begin{itemize}
\item We consider two kinds of sequences of observations: 
a) Independent and identically distributed (i.i.d) observations in $\mathbb{C}^{n}$ denoted by $\mathbf{s} = (\mathbf{x}_{1},\mathbf{x}_{2},\cdots\mathbf{x}_{i}\cdots)$ and b) Observations mapped to RKHS using RFF denoted by $\bm{\zeta}=(\bm{\hat{\Phi}}_{1},\bm{\hat{\Phi}}_{2},\cdots\bm{\hat{\Phi}}_{i}\cdots)$.

\item We denote the linear inner-product space in $\mathbb{C}^{n}$ as $\mathcal{X}$. Furthermore, we denote the space spanned by $\bm{\Phi}_{1}$ as $\mathcal{H}$, which can be considered as an extension space of $\mathcal{X}$.

c) The sequence of actual, and predicted labels are denoted by $\mathbf{g}=({g}_{1},{g}_{2},\cdots{g}_{i}\cdots)$ and $\hat{\mathbf{g}}=(\hat{g}_{1},\hat{g}_{2},\cdots\hat{g}_{i}\cdots)$, respectively.

d) We assume that the LSTM is a system that inputs sequences and outputs the labels asymptotically correct for some time-index $i>T$, i.e. 
$\text{Pr}(\hat{g}_{i}=g_{i})\to 1, \forall i>T$, where T is arbitrarily large.

\item Furthermore, we assume two hypothetical cases: 

i) $\mathbf{s}$ is the input to an LSTM  based predictor.

ii) $\bm{\zeta}$ is the input to the same LSTM considered in i).
\end{itemize}

Based on these assumptions we proceed to formulate the following theorem.
\begin{thm}
For a given LSTM network, the likelihood of correct detection of a sequence of labels given the mapped sequence $\bm{\zeta}$, is greater than the likelihood of corresponding correct detection given $\mathbf{s}$ for some time-index $i>T$, where $T$ is large enough.
\end{thm}

\begin{proof}
We provide a proof by contradiction. First we assume the contrapositive, i.e.
\small
\begin{gather}
\label{eq1}
p(\mathbf{\hat{g}}=\mathbf{g}|\mathbf{s})>p(\mathbf{\hat{g}}=\mathbf{g}|\bm{\zeta}).
\end{gather}
\normalsize
Further for $t>T$, we can assume $p(\mathbf{g})\sim\mathcal{N}(\mathbf{\hat{g}},\epsilon^2)$, and $\epsilon^{2}\to 0$ (where $\epsilon^{2}$ is the approximation-error energy) and 
$\mathcal{N}(\mu,\sigma^2)$ denotes a Gaussian distribution with mean $\mu$ and variance $\sigma^2$ (which is a soft approximation of $g$ in the neighborhood of $\hat{g}$).

Under this assumption, one can write (\ref{eq1}) as
\small
\begin{gather}
\label{eq2}
\int_{\mathbf{g}}\frac{p(\mathbf{s}|\mathbf{g})p(\mathbf{g})}{p(\mathbf{s})}d\mathbf{g}
>\int_{\mathbf{g}}\frac{p(\bm{\zeta}|\mathbf{g})p(\mathbf{g})}{p(\mathbf{\bm{\zeta}})}d\bm{\zeta}.
\end{gather}
\normalsize
Taking the logarithm of both sides, applying Jensen's inequality, and the $\min$ operation, one can re-express (\ref{eq2}) as
\small
\begin{eqnarray}
\int\limits_{\mathbf{g}}[\log p(\mathbf{s}|\mathbf{{g}})-\log(p(\mathbf{s}))]\geq\int\limits_{\mathbf{g}}[\log p(\bm{\zeta}|\mathbf{{g}})-\log(p(\bm{\zeta}))]\ .
\end{eqnarray}
\normalsize
Taking the expectation with respect to $\mathbf{s}$ on both sides (denoted by $\mathbb{E}_{\mathbf{s}}[\cdot]$), and letting $-\mathbb{E}_{\mathbf{s}}[\log(p(\mathbf{s})]=H(\mathbf{s})$, where $H(\cdot)$ denotes the Shannon entropy, we can rewrite the above equation as\footnote{For notational convenience, we have overloaded $\bm{\zeta}(\mathbf{x})$ as simply $\bm{\zeta}$}
\small
\begin{eqnarray}
\mathbb{E}_{\mathbf{s}}\Big[\int\limits_{\mathbf{g}}\log[p(\mathbf{s}|\mathbf{{g}})]\Big]+H(\mathbf{s})\geq\mathbb{E}_{\mathbf{s}}\Big[\int\limits_{\mathbf{g}}\log[p(\bm{\zeta}|\mathbf{{g}}]\Big]+H(\bm{\zeta}) .
\end{eqnarray}
\normalsize
Using the fact that
\small
\begin{gather}
\max_{p(\mathbf{s})}H(\mathbf{s})>\max_{p(\mathbf{s})}H(\bm{\zeta}),
\end{gather}
\normalsize
and defining $\theta>0$, such that $\theta=\max\limits_{p(\mathbf{s})}[H(\mathbf{s})-H(\bm{\zeta})]$, we get
\small
\begin{gather}
\mathbb{E}_{\mathbf{s}}\Big(\int_{\mathbf{g}}\log[p(\mathbf{s}|\mathbf{{g}})]\Big)\geq\mathbb{E}_{\mathbf{s}}\Big(\int_{\mathbf{g}}\log[p(\bm{\zeta}|\mathbf{{g}})]\Big)-\theta.
\end{gather}
\normalsize
In other words,
\small
\begin{gather}
\mathbb{E}_{\mathbf{s}}\Big(\int_{\mathbf{g}}[\log p(\bm{\zeta}|\mathbf{{g}})-\log p(\mathbf{s}|\mathbf{{g}})]\Big)\leq\theta.
\end{gather}
\normalsize
For $i>T$, the measure of $\mathbf{g}$ is concentrated around $\hat{\mathbf{g}}$; hence, under this assumption the above equation can be re-expressed as
\small
\begin{gather}
0<\mathbb{E}_{\mathbf{s}}\Big([\log p(\bm{\zeta}|\mathbf{g}=\hat{\mathbf{{g}}})-\log p(\mathbf{s}|\mathbf{g}=\hat{\mathbf{{g}}})]\Big)\leq\theta.
\end{gather}
\normalsize
Hence, we have
\small
\begin{gather}
\log p(\bm{\zeta}|\mathbf{g}=\hat{\mathbf{{g}}})>\log p(\mathbf{s}|\mathbf{g}=\hat{\mathbf{{g}}}),
\end{gather}
\normalsize
or,
\small
\begin{gather}
\log p(\mathbf{g}=\hat{\mathbf{{g}}}|\bm{\zeta})>{\log p(\mathbf{g}=\hat{\mathbf{{g}}}|\mathbf{s})}+\log\Big[\frac{p(\mathbf{s})}{p(\bm{\zeta})}\Big].
\end{gather}
\normalsize
Using the cosine transform-relation of the random variables from $\mathbf{x}$ to $\bm{\zeta}$, we conclude $\frac{p(\mathbf{s})}{p(\bm{\zeta})}\geq1$ (as the Lebesgue measure of $\mathcal{X}$ is less than the Lebesgue measure of $\mathcal{H}$\footnote{It is noted that while there are several extensions from $\mathbf{x}$ possible, the RFF based extension outperforms classical polynomial kernel based extensions \cite{rmitra2020}}), and thus we reach the following inequality:
\small
\begin{gather}
{\log p(\mathbf{g}=\hat{\mathbf{{g}}}|\mathbf{s})}<{\log p(\mathbf{g}=\hat{\mathbf{{g}}}|\bm{\zeta})}.
\end{gather}
\normalsize
This yields a contradiction from our original claim in (\ref{eq1}), which concludes the proof.
\end{proof}
\section{Proposed distribution dependent RFFs}
The RFF based DNNs require a large number of RFFs to gain an accurate approximation of an RKHS. In this context, a distribution-dependent RFF is proposed in this section, which achieves lower approximation error as compared to classical RFFs for a given number of RFF dimensions, and hence significantly lowers the computational complexity required for achieving a given error floor. 

Indexing the incoming observations by $j$, we can update the following Parzen estimate of the p.d.f of observations, denoted by $\hat{p}(\mathbf{x})$.
\begin{gather}
\label{eqnnn}
\hat{p}(\mathbf{x}):= \Big(\frac{j-1}{j}\Big)\hat{p}(\mathbf{x}) + \frac{1}{j}\mathcal{K}_{\lambda}(\mathbf{x}-\mathbf{x}_{j}),
\end{gather}
where the spread parameter for kernel density estimation, $\lambda$, is drawn from Silverman's rule \cite{silverman1986density}.
Consequently, the mean of the RFF can be readily derived, using a moving-average estimator as follows:
\begin{gather}
\bm{\mu}_{\hat{\bm{\Phi}}}:=\nu\bm{\mu}_{\hat{\bm{\Phi}}}+\underbrace{\int_{\mathbf{x}}\hat{\Phi}(\mathbf{x})\hat{p}(\mathbf{x})}_{\bm{\mathcal{S}}},
\end{gather}
where $\bm{\mu}_{\hat{\bm{\Phi}}}$ denotes the mean of the RFF, estimated by the moving-average estimator, and $\nu\in[0,1]$ is the forgetting factor. However, from (\ref{eqnnn}) one can adapt $\mathcal{S}$ as
\begin{gather}
\bm{\mathcal{S}}:=\frac{i-1}{i}\bm{\mathcal{S}} +\frac{1}{i}\frac{1}{M}\int_{\mathbf{x}}\sum\limits_{j=1}^{M}\mathcal{K}_{\lambda}(\mathbf{x}-\mathbf{x}_{j})\bm{\hat{\Phi}}(\mathbf{x})d\mathbf{x}.
\end{gather}
Denoting $\bm{\mathcal{S}}=\{\mathcal{S}_{i}\}_{i=1}^{n_{G}}$ as a vector, we get:
\begin{gather}
\label{DD_RFF}
\mathcal{S}_{i}:=\frac{i-1}{i}\mathcal{S}_{i} +\exp\Big(-\frac{\lambda^2\bm{\omega}_{i}^{T}\bm{\omega}_{i}}{2}\Big)\times\\\nonumber\frac{1}{i}\frac{\sqrt{2\pi\lambda^2}}{M}\sum\limits_{j=1}^{M}
\cos(\bm{\omega}_{i}^{T}\mathbf{x}_{j}+b_{n_{G}}),  
\end{gather}
where $M$ is the size of the considered batch. This gives us a new smoothed RFF that can present potentially useful features for low-complexity DNN architectures.
\section{Proof of viability of distribution dependent RFFs}
In the Section. III, we demonstrated the viability of RFF based LSTMs in terms of the mis-classification error/the ``hit-or-miss" cost function. In this section, we prove that, compared to classical RFFs, the proposed distribution dependent RFFs provide a better approximation to the RKHS by reducing the approximation-error (which is a ``soft" error-metric). We state this claim in the form of the following two theorems.
\begin{thm}
The mapping given by the distribution-dependent RFF map is closed under the same RKHS $\mathcal{H}$ as the classical RFF. Furthermore, distribution dependent RFFs formulated in (\ref{DD_RFF}) deliver an improved approximation of RKHS.
\end{thm}
\begin{proof}
We begin by recalling that for a fixed number of $n_{G}$ dimensions, the approximation error $\epsilon$ for classical RFF is 
given by $n_{G}\sim O(\epsilon^{-2}\log\epsilon^{-2})$ \cite{qin2017random}. However, we also note that $\epsilon$ is not only a function of dimension $n_{G}$, but also a function of the set of $\bm{\omega}_{i}$. In other words, the approximation error depends on how much the samples of $\omega_{i}$ deviate from a Gaussian distribution; particularly, when $n_{G}$ is not high enough to converge to the desired Gaussian distribution. Hence, in the sequel, we denote the approximation error as $\epsilon_{n_{G}}(\bm{\omega})$, where $\bm{\omega}$ denotes an approximate continuum of $\{\bm{\omega}_{i}\}_{i=1}^{n
_{G}}$.

From (\ref{DD_RFF}), we can approximate each component of the proposed RFF as follows:
\small
\begin{gather}
\frac{\sqrt{2\pi\lambda^2}}{M}\sum\limits_{j=1}^{M}\int\limits_{\bm{\omega}}\exp\Big({-\frac{\lambda^2\|\bm{\omega}\|^2}{2}}\Big)\cos(\bm{\omega}^{T}\mathbf{x}_{j}+b)d\bm{\omega}.
\end{gather}
\normalsize
Noting the fact that the RFF can be expressed as the sum of an element in RKHS $\mathcal{H}$ with some error $\epsilon_{n_{G}}$, the above equation can be re-expressed as
\small
\begin{gather}
\frac{\sqrt{2\pi\lambda^2}}{M}\sum\limits_{j=1}^{M}\int\limits_{\bm{\omega}}\exp\Big({-\frac{\lambda^2\|\bm{\omega}\|^2}{2}}\Big)\Big[\mathcal{H}_{j}+\epsilon_{n_{G_{j}}}(\bm{\omega})\Big]d\bm{\omega},
\end{gather}
\normalsize
which can be simplified as
\small
\begin{gather}
2\pi\lambda^2\mathcal{H}_{j}+\\\nonumber\frac{\sqrt{2\pi\lambda^2}}{M}\sum\limits_{j=1}^{M}\underbrace{\int\limits_{\bm{\omega}}\exp\Big({-\frac{\lambda^2\|\bm{\omega}\|^2}{2}}\Big)\epsilon_{n_{G_{j}}}(\bm{\omega})d\bm{\omega}}_{\epsilon_{n_{G}}^{(1)}(\bm{\omega})}.
\end{gather}
\normalsize
From the above equation, we make the following observations:
\begin{itemize}
\item Applying Parseval's theorem, one can easily note that the energy of $\epsilon_{n_{G}}^{(1)}(\bm{\omega})$ is lower than $\epsilon_{n_{G}}(\bm{\omega})$, as the $\exp\Big({-\frac{\lambda^2\|\bm{\omega}\|^2}{2}}\Big)$ part performs a low-pass filtering, and attenuates higher-magnitude ``frequencies" (or $\bm{\omega}$).
\item One may note that the desired component, which is apart from $\epsilon_{n_{G}}(\bm{\omega})$, lies in the RKHS $\mathcal{H}$ (though it is scaled by $2\pi\lambda^2$, we use the closure property of $\mathcal{H}$ under scaling to prove closedness of the feature-map under the proposed distribution-dependent RFF).
\end{itemize}
\end{proof}
\begin{thm}
For the same error floor $\epsilon$, the number of RFFs required by the distribution dependent RFFs is lower by a factor $\mathcal{O}(\big(\frac{\lambda^{2}}{2\pi}\big)^{n})$.
\end{thm}
\begin{proof}
It is noted that upon using the distribution dependent RFFs, and upon invoking the Cauchy-Schwarz inequality, the number of RFF dimensions required for an error floor $\mathcal{K}\epsilon$, is written as follows:
\begin{gather}
n_{G}^{1}=
\mathcal{C}\mathcal{K}^{-2}\epsilon^{-2}\log(\mathcal{K}^{-2}\epsilon^{-2}),
\end{gather}
where $\mathcal{C}$ is an arbitrary constant, and $\mathcal{K}=\Big[\frac{\lambda^2}{2\pi}\Big]^{\frac{n}{2}}<1$ is a scaling constant making the error floors of RFF and distribution-dependent RFF same to facilitate comparison. Similarly, for distribution-dependent RFFs, the dimensions $n_{G}$ can be quantified as
\begin{gather}
n_{G}^{2}=\mathcal{C}\epsilon^{-2}\log(\epsilon^{-2}).
\end{gather}
Hence, the ratio may be expressed as
\begin{gather}
\frac{n_{G}^{1}}{n_{G}^{2}} = \frac{\lambda^{2n}}{(2\pi)^{n}}\Big[1+\frac{\epsilon^{2}n_{G}n}{4}\log\Big(\frac{\lambda^2}{2\pi}\Big)\Big],
\end{gather}
Noting that $\epsilon^{2}\approx 0$ for sufficiently large RFF dimensions, we can write
\begin{gather}
\frac{n_{G}^{1}}{n_{G}^{2}} \approx \frac{\lambda^{2n}}{(2\pi)^{n}}.
\end{gather}
It can be noted that since $\lambda<1$, which proves the desired result.
\end{proof}
\section{Architecture of RFF based LSTM}
Considering the results presented in Theorem 2, wherein the viability of the proposed RFF based features to an LSTM for sequence detection is established,  we describe the neural network architecture considered in this work. For illustration purposes, the proposed architecture is shown in Fig. \ref{fig0}.
\begin{figure*}[!htbp]
\centering
  \includegraphics[width=0.8\linewidth,height=1\linewidth]{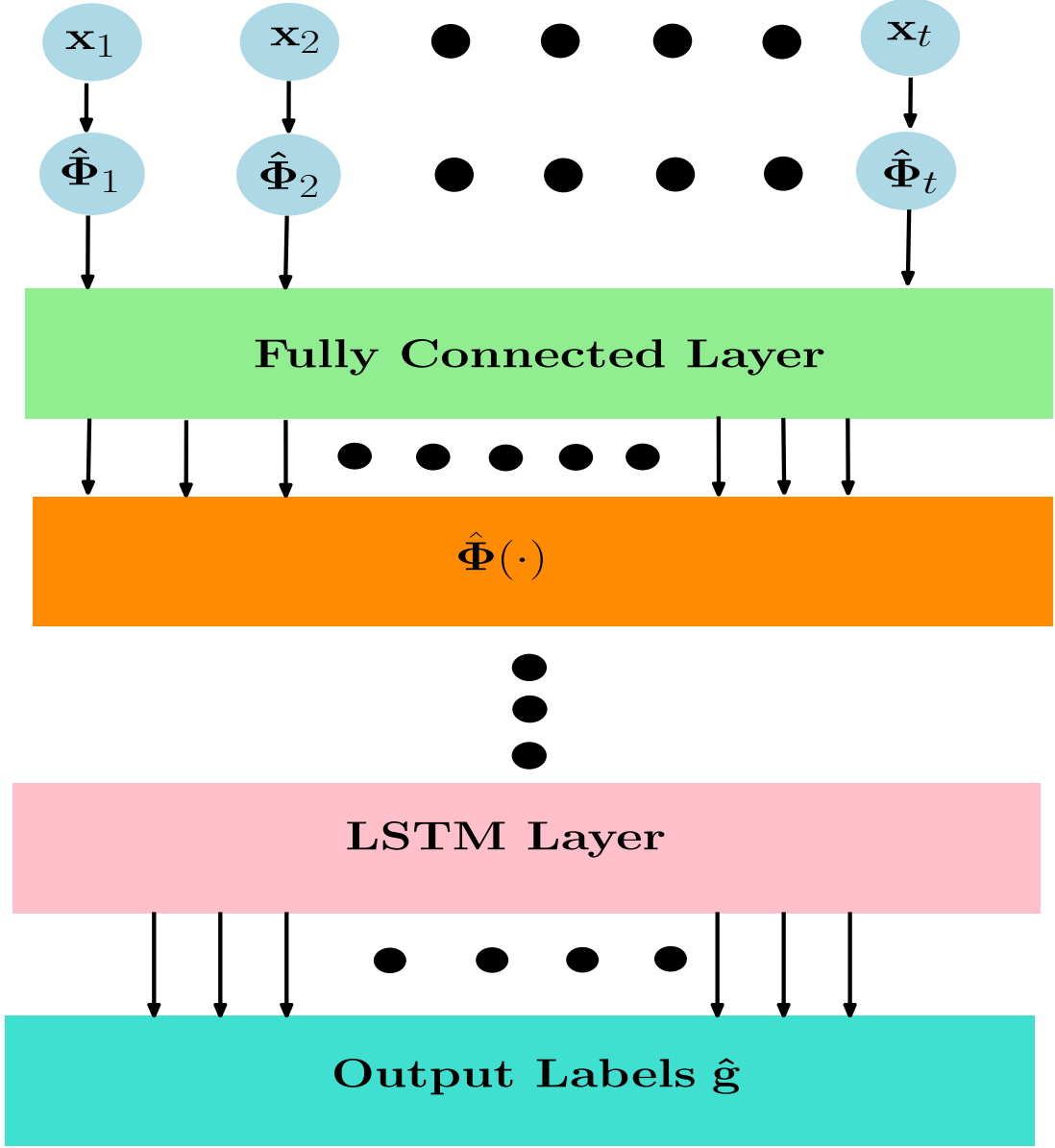}\\
  \caption{Depiction of LSTM based prediction using RFF.}\label{fig0}
\end{figure*}
As detailed in Theorem 2, the input is mapped to an approximate RKHS using RFFs as outlined in (\ref{RFF}) or (\ref{DD_RFF}). Next, there are optional fully connected layers with a subsequent RFF transformation. The cascade of consecutive mapping renders the overall mapping to a sequence of RKHSs and the result in Theorem 2, can be readily utilized by replacing $\mathcal{H}$ with the last RKHS. Lastly, the observations in RKHS are presented as an input to the LSTM layer for prediction of the NLOS/LOS labels. Notably, the benefits of the proposed mapping prior to presentation to the LSTM layer have been highlighted in Theorem 2, which indicates that the posterior is more ``peaked" given RFFs as input to the LSTM, as compared to the indigenous observations. These steps can be summarized in the following set of equations:
\begin{gather}
\label{proposed_LSTM}
\bm{f}_{t} = \sigma_{g}(\bm{\mathcal{B}}^{(1)}\mathbf{\hat{\Phi}}(\mathbf{x}_{t})+\bm{\mathcal{B}}^{(2)}\bm{\gamma}_{t-1}+\bm{b}_{f})\\\nonumber
\bm{i}_{t} = \sigma_{g}(\bm{\mathcal{A}}^{(1)}\mathbf{\hat{\Phi}}(\mathbf{x}_{t})+\bm{\mathcal{A}}^{(2)}\bm{\gamma}_{t-1}+\bm{b}_{i})\\\nonumber
\bm{o}_{t} = \sigma_{g}(\bm{\mathcal{D}}^{(1)}\mathbf{\hat{\Phi}}(\mathbf{x}_{t})+\bm{\mathcal{D}}^{(2)}\bm{\gamma}_{t-1}+\bm{b}_{o})\\\nonumber
\bm{\chi}_{t} = \bm{\chi}_{t-1}\circ\bm{f}_{t}+\bm{i}_{t}\circ\sigma_{c}(\bm{\mathcal{F}}^{(1)}\mathbf{\hat{\Phi}}(\mathbf{x}_{t})+\cdots\\\nonumber\bm{\mathcal{F}}^{(2)}\mathbf{h}_{t-1}+\bm{b})\\\nonumber
\bm{\gamma}_{t}=\bm{o}_{t}\circ\sigma_{h}(\bm{\chi}_{t}),
\end{gather}
where $\circ$ denotes the Hadamard product, $\sigma_{g}(\cdot)$ is a sigmoid activation function, $\sigma_{c}(\cdot)$ and $\sigma_{h}(\cdot)$ denote the $\tanh(\cdot)$ activation function, and  $\bm{f}_{t},\bm{i}_{t},\bm{o}_{t},\bm{\chi}_{t},\bm{\gamma}_{t}\in\mathbb{R}^{n_{h}}$ where $n_{h}$ denotes number of hidden nodes. Furthermore, 
$\bm{f}_{t}$ denotes the forget-gate of LSTM with weights $\bm{\mathcal{B}}^{(1)}\in\mathbb{R}^{n_{h}\times n_{G}},\bm{\mathcal{B}}^{(2)}\in\mathbb{R}^{n_{h}\times n_{h}}$ and bias $\bm{b}_{f}\in\mathbb{R}^{n_{h}}$). $\bm{i}_{t}$ denotes the input gate with weights $\bm{\mathcal{A}}^{(1)}\in\mathbb{R}^{n_{h}\times n_{G}},\bm{\mathcal{A}}^{(2)}\in\mathbb{R}^{n_{h}\times n_{h}}$ and bias $\bm{b}_{i}\in\mathbb{R}^{n_{h}}$. $\bm{o}_{t}$ denotes the output gate with weights $\bm{\mathcal{D}}^{(1)}\in\mathbb{R}^{n_{h}\times n_{G}},\bm{\mathcal{D}}^{(2)}\in\mathbb{R}^{n_{h}\times n_{h}}$ and bias $\bm{b}_{o}\in\mathbb{R}^{n_{h}}$. $\bm{\chi}_{t}$ is the sum of the gating of the previous parameter-value with the forget-gate and the gating between input nodes and input-layer
(parameterized by weights $\bm{\mathcal{F}}^{(1)}\in\mathbb{R}^{n_{h}\times n_{G}},\bm{\mathcal{F}}^{(2)}\in\mathbb{R}^{n_{h}\times n_{h}}$ and bias $\bm{b}\in\mathbb{R}^{n_{h}}$), which is passed through the activation function $\sigma_{c}(\cdot)$ and gated with the output gate $\bm{o}_{t}$ to obtain the next-state $\bm{\gamma}_{t}\in\mathbb{R}^{n_{h}}$. The aforementioned weights and biases are optimized using the backpropagation algorithm; however, we do not encounter the vanishing/exploding gradient problem for time-series prediction. In addition, the mapping to an approximate RKHS $\bm{\hat{\Phi}}(\cdot)$ (as derived in (\ref{proposed_LSTM})), prior to input to LSTM increases the accuracy of predictions as inferred from Theorem. 2.
 
In the next section, we present simulations, which validate the viability of the proposed RFF based LSTM for two problems. First, LOS/NLOS classification over an outdoor WINNER II channel. Next, the suitability of the proposed distribution-dependent RFF for LDPC decoding is demonstrated by realistic simulations over nonlinear VLC channels. 
\section{Simulations}
In this section, we present simulations to validate the paradigm of RFF based learning for LOS/NLOS based classification and message-passing based detection for LDPC decoding. From the simulations presented below, one can observe significant gains for various nonlinear classification problems when using RFF-approximations of an RKHS, compared to using the indigenous observations.
\subsection{RFF-based LSTM for LOS/NLOS classification}
In this subsection, simulations are presented for LOS/NLOS identification in outdoor communication systems. The simulation parameters are summarized in Table. \ref{eval_table}. We consider various outdoor WINNER II channel scenarios, wherein there is a single base-station, and the receiver moves from an initial location using a random-walk mobility model. The antenna height at the transmitter is assumed to be 4m, while the mobile stations were assumed to move along a trajectory in the horizontal plane drawn from a 2D random-walk mobility model. The OFDM standard assumed is IEEE 802.11ax with a guard band of 3.2 $\mu$s. The complex channel-estimates at the receiver are transformed by concatenating real and imaginary components prior to RFF mapping.
In this sections, we present the following two comparison-cases: a) the case in which we present $\mathbf{x}$ directly to the LSTM layer and b) the case in which $\mathbf{x}$ is mapped to an RKHS using RFF layer(s) prior to the LSTM layer.
The candidate neural networks were trained with the initial location of the receiver as mentioned in Scenario I of Table \ref{eval_table}\footnote{From our simulations, we also note that the proposed approach works for an arbitrary initial point; however, we show some "extreme" skewed cases for brevity}. It is worthwhile to mention that the initial location of the user for Scenario I was chosen heuristically such that there are almost equal number of examples of LOS and NLOS observations, as the user moves along the aforementioned random-walk based trajectory. Consequent to training, the candidate neural networks were tested on the testing-observations derived from Scenario II (which has more LOS labels due to the initial location being near to the base station), and Scenario III (which has more NLOS labels due to the initial location being far away from the base station). The testing $F_{1}$-score is plotted for the C1 and D1 outdoor scenarios. It is observed that the gains in $F_{1}$-score performance becomes more prominent in the low training-data regime, as seen from Fig. \ref{fig3} and Fig. \ref{fig4}, which makes  distribution dependent RFFs better suited for rapidly fluctuating outdoor scenarios. Further, from the simulated receiver operating characteristics (ROC), which are plotted in Fig. \ref{fig5} and Fig. \ref{fig51} for 400 training samples, better performance is obtained from the RFF based LSTM than the generic LSTM, which is in line with the gains promised by the analytical results derived previously.
\begin{table*}
\caption{Simulation Parameters}
\begin{center}
\begin{tabular}{|c|c|}
\hline
Outdoor channel model & Winner-II \\
\hline
Mobility model & Random walk  \\
\hline
Guard Band & 3.2 $\mu$s  \\ 
\hline
OFDM Standard & IEEE 802.11ax  \\
\hline
Antenna Height at Tx & 4m \\
\hline
Coordinates of base-station & (50,150)\\
\hline
(Training): Initial Rx coordinates C2 & (200,120)\\ 
\hline
(Training): Initial Rx coordinates C1 & (60,200)\\ 
\hline
(Testing): Initial Rx coordinates C2& (300,250),(60,200)\\
\hline
(Testing): Initial Rx coordinates C1& (300,250),(450,500)\\
\hline
[Urban scenario]: LOS/NLOS probability for C1 & \cite[eq. (1)]{JG},\cite[Table. 4.7]{bultitude20074}\\
\hline
[Suburban scenario]: LOS/NLOS probability for C2& \cite[eq. (2)]{JG},\cite[Table. 4.7]{bultitude20074}\\
\hline
Number of RFF layers& 1\\
\hline
$\mathbf{x}$ & Sequence of Channel-estimates\\
\hline
RFF dimensions for input layer & 200 \\
\hline
Number of LSTM layers & 1 \\
\hline
Number of Hidden nodes for LSTM layer & 50\\
\hline
Considered WINNER II Scenarios (Suburban) & C1, D1 \\
\hline
Test-data sequence length & 6000\\
\hline
\end{tabular}
\label{eval_table}
\end{center}
\end{table*}
\subsection{LDPC decoding for VLC}
In this subsection, we describe our methodology for LDPC decoding in a nonlinear VLC channel. We assume LOS VLC channel modelled by a Lambertian radiation pattern \cite{ghassemlooy2017visible,mitra2017precoded}, with a memory Rapp LED nonlinearity (which is widely used for modelling a white light-emitting diode (LED) \cite{elgala2010led,jain2019klms}). The overall system model at the $i^{th}$ time-instant can be written as
\begin{gather}
\mathbf{y}_{i} = h_{i}f(\mathbf{x}_{i}+\alpha \mathbf{x}_{i-1}) + \mathbf{n},
\label{sysmod}
\end{gather}
where $\mathbf{n}\sim\mathcal{N}(0,\sigma_{n}^{2}\mathbb{I})$, with $\sigma_{n}^{2}$ denoting the overall variance of the additive noise which accounts for the overall effect of shot-noise and ambient noise at the photodetector. Moreover, $x_{i}$ denotes encoded independent and identically distributed (i.i.d) on-off keying (OOK) transmissions, and $f(\cdot)$ denotes the LED transfer characteristic modelled as an AM-AM Rapp nonlinearity. The encoding is performed according to the 802.11n LDPC generator-matrix. From (\ref{sysmod}), one can observe that the nonlinearity $f(\cdot)$ warps/distorts the transmitted codewords $\mathbf{x}$, which alters their algebraic structure, and causes errors in LDPC decoding. Hence, based on the previous discussion, it is proposed to learn the bits based on a detector trained on $\bm{\hat{\Phi}}(\mathbf{y})$, where the RFF dimensions are equal to the codeword-length\footnote{It is also possible to up-convert to higher dimensions using an RFF and then down-convert using an autoencoder. Though this dual-conversion may have performance benefits, it is computationally complex, and hence we focus our attention on the single-layer case.}. In fact, using the Representer theorem, (\ref{sysmod}) can be re-written as
\begin{gather}
\mathbf{y} = <k_{x},\mathbf{x}>_{\mathcal{H}} + \mathbf{n},
\end{gather}
where $k_{x}$ is an operator in RKHS $\mathcal{H}$.
Using the completeness of the RKHS $\mathcal{H}$, there exists an operator $k_{y}$ such that
\begin{gather}
<k_{y},\mathbf{y}>_{\mathcal{H}} \approx \mathbf{x} + <k_{y},\mathbf{n}>_{\mathcal{H}},
\end{gather}
which can be modelled as an AWGN channel with noise-variance equal to$$\text{var}[<k_{y},\mathbf{y}>_{\mathcal{H}}]=<k_{y},k_{y}>_{\mathcal{H}}\sigma_{n}^{2}.$$ 
Given the theme of this work, the following approximation of RKHS is utilized
$$<k_{y},\mathbf{y}>_{\mathcal{H}}\approx\bm{\Omega}^{T}\bm{\hat{\Phi}}(\mathbf{y}).$$
Next, a hypothesis, denoted by $\bm{\Omega}$, is trained on 
$\bm{\hat{\Phi}}(\mathbf{y})$, which optimizes the following quadratic loss-function $$\|\bm{\Omega}^{T}\bm{\hat{\Phi}}(\mathbf{y})-\text{sign}(\mathbf{y})\|_{2}^2$$
Channel-decoding is performed using message-passing over a Tanner graph representation of the parity-check matrix \cite{johnson2010iterative}. We denote the graph-neighorbood of node $k$ as $\mathcal{B}_{k}$ (which is the set of points incident on node $k$ in the Tanner-graph apart from $k$ itself). Additionally, for the $j^{th}$ bit, the log-likelihood-messages from the bit-nodes to codewords, and the message from codewords to bit-nodes are denoted as $m_{b}(j)$ and $m_{c}(j)$ respectively.  Lastly, we denote the length of the bit-string $\mathbf{b}$ as $B$ and the size of the encoded codeword as $C$. The algorithm is summarized in Algorithm 1.
\begin{algorithm}
	\caption{Message-Passing in RKHS} 
	\begin{algorithmic}[1]
	\State Initialize $\text{MAXITER}$.
	\For {$k=1,2,\ldots C$}
		\State Initialize LLR: $m_{c}(k)=\frac{-2 y(k)}{\sigma_{n}^{2}}$. 
		\EndFor
	\While{$\text{cnt}<\text{MAXITER}$}
	\Repeat 
		
			\For {$j=1,2,\ldots,B$}
				\State $m_{b}(j) = \log\Bigg[\frac{1+\prod\limits_{k\in\mathcal{B}_{j}}\tanh(m_{c}(k))}{1-\prod\limits_{k\in\mathcal{B}_{j}}\tanh(m_{c}(k))}\Bigg]$ 
			\EndFor
			\For {$j=1,2,\ldots,B$}
				\State $m_{b}^{'}(j) = \sum\limits_{j\cup j\in\mathcal{B}_{j}}m_{c}(j)-\frac{2y(k)}{\sigma_{n}^{2}}$ 
\State $\hat{b}(j)=\text{sign}(m_{b}^{'}(j))$.
			\EndFor
			\State Calculate syndome $\bm{e} = \mathbf{H}\hat{\mathbf{b}}$.
		\For {$j=1,2,\ldots,C$}
				\State Update: $m_{b}(j) = \sum\limits_{j\in\mathcal{B}_{j}}m_{c}(j)-\frac{2y(k)}{\sigma_{n}^{2}}$ 
			\EndFor	
			\Until{$\bm{e}=\bm{0}$}
			\State  Estimate a map $\bm{\Omega}$ such that
			minimizes $\|\bm{\Omega}^{T}\bm{\hat{\Phi}}(\mathbf{y})-\text{sign}(\mathbf{y})\|^{2}$.
			\State Estimate variance: $\hat{\sigma}^{2} = \text{var}[\bm{\Omega}^{T}\bm{\hat{\Phi}}(\mathbf{y})]$.
			\For {$k=1,2,\ldots C$}
		\State Re-initialize LLR: $m_{c}(k)=m_{c}(k)+\frac{-2 \bm{\Omega}^{T}\bm{\hat{\Phi}}(\mathbf{y})}{\hat{\sigma}^{2}}$. 
		\EndFor
		\State Update regressors: $\mathbf{y} = \bm{\Omega}^{T}\bm{\hat{\Phi}}(\mathbf{y})$.
			\EndWhile
	\end{algorithmic} 
\end{algorithm}
\begin{figure}[!htbp]
\centering
  \includegraphics[width=1\linewidth,height=7cm]{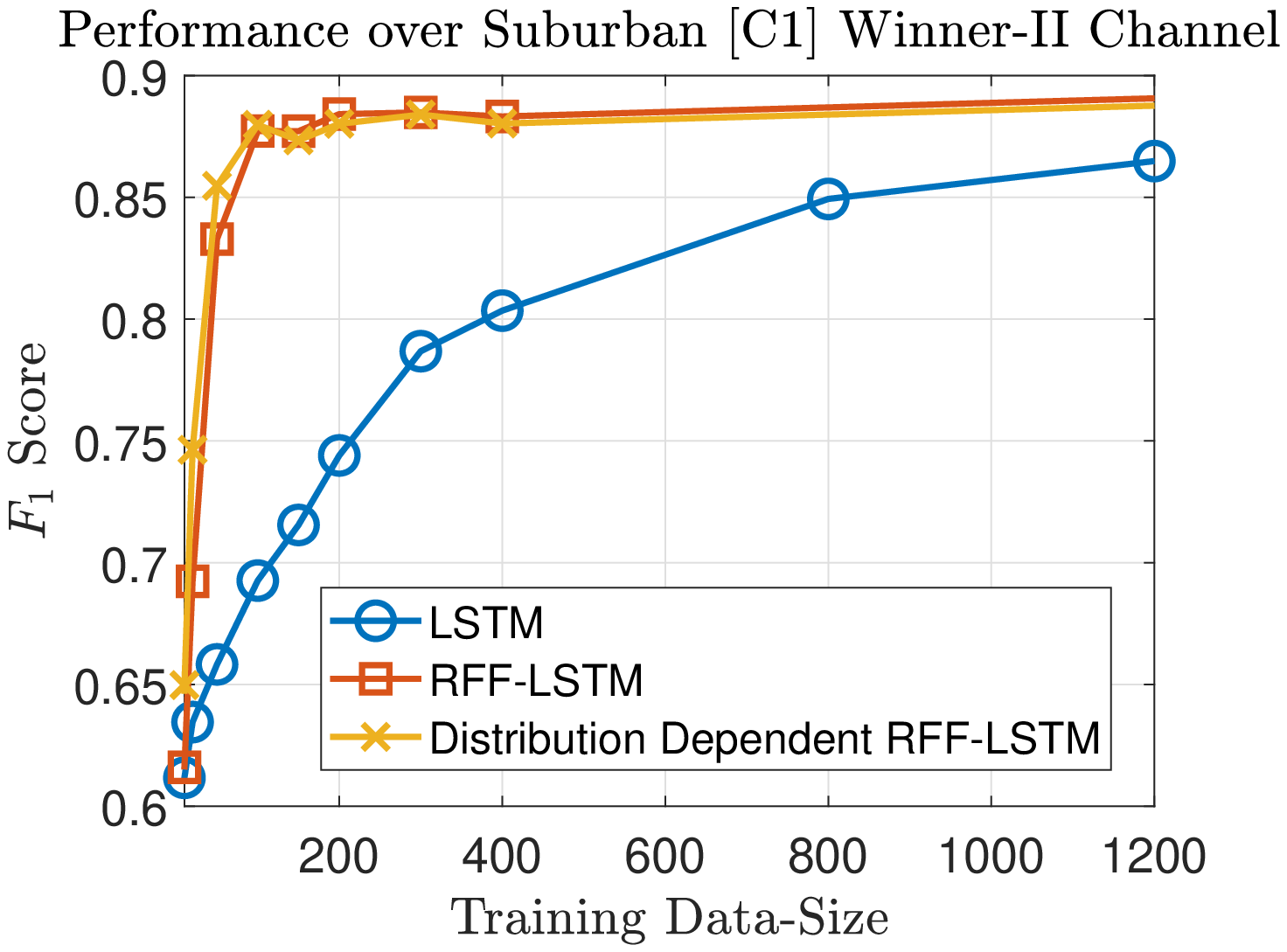}\\
  \caption{F1-score obtained in the LOS/NLOS classification for LSTM, RFF-LSTM, and the Distribution dependent RFF-LSTM over urban scenario C1.}\label{fig3}
\end{figure}
\begin{figure}[!htbp]
\centering
  \includegraphics[width=1\linewidth,height=7cm]{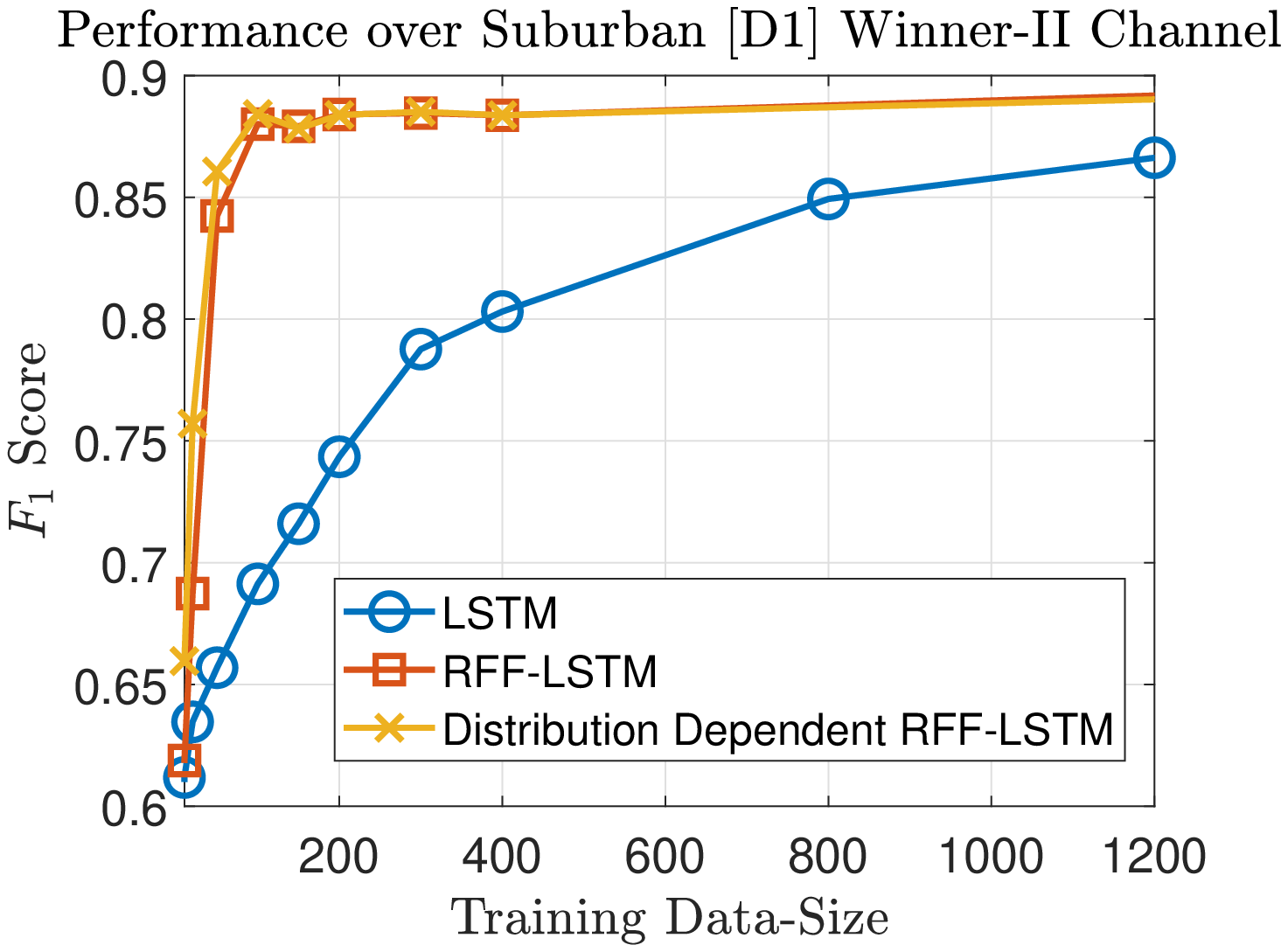}\\
  \caption{F1-score obtained in the LOS/NLOS classification for LSTM, RFF-LSTM, and the Distribution dependent RFF-LSTM over suburban scenario D1.}\label{fig4}
\end{figure}
\begin{figure}[!htbp]
\centering
  \includegraphics[width=1\linewidth,height=7cm]{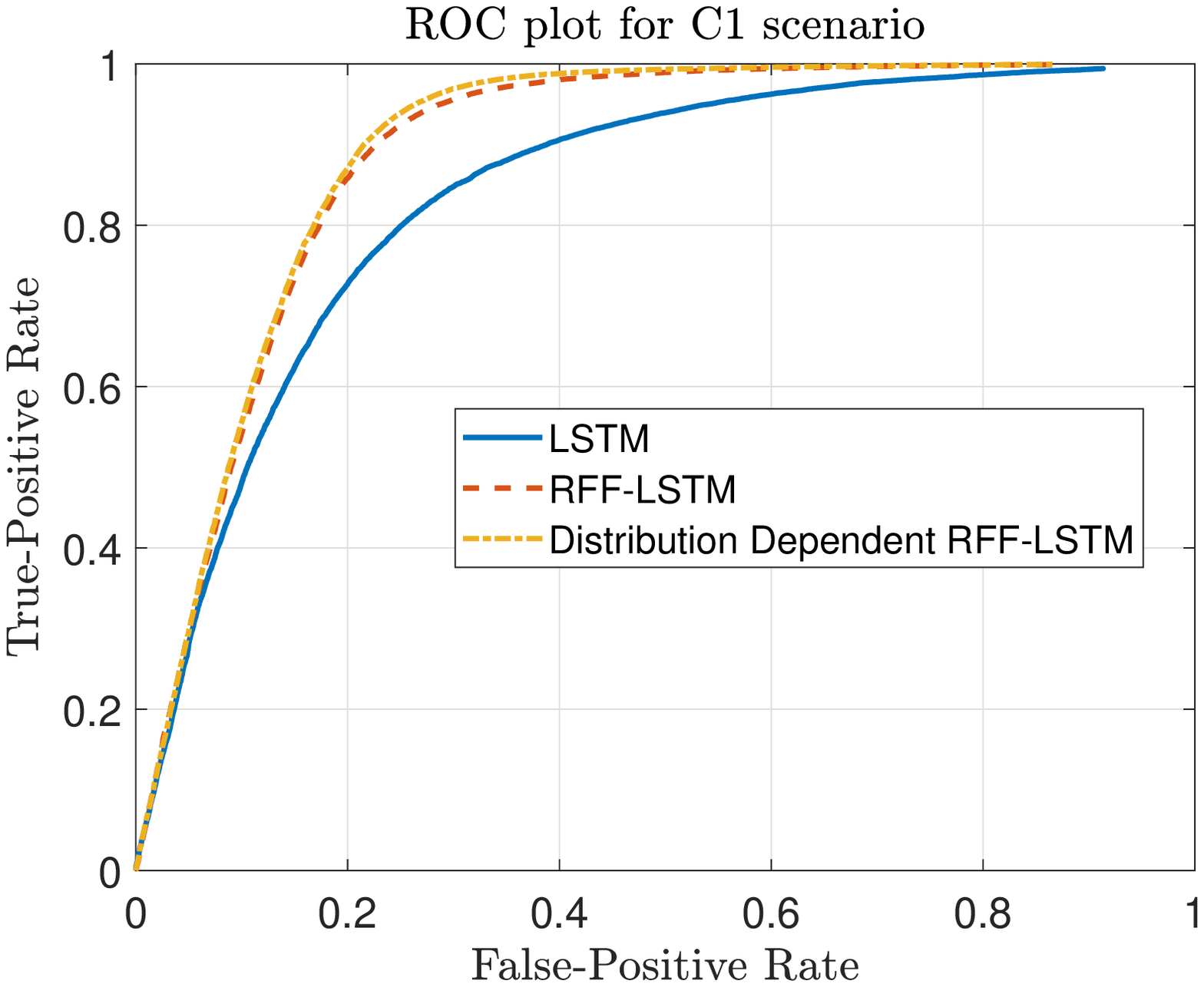}\\
  \caption{ROC curves obtained in the LOS/NLOS classification for LSTM, RFF-LSTM, and the Distribution dependent RFF-LSTM over suburban scenario C1 with a training-size of 400 observations.}\label{fig5}
\end{figure}
\begin{figure}[!htbp]
\centering
  \includegraphics[width=1\linewidth,height=7cm]{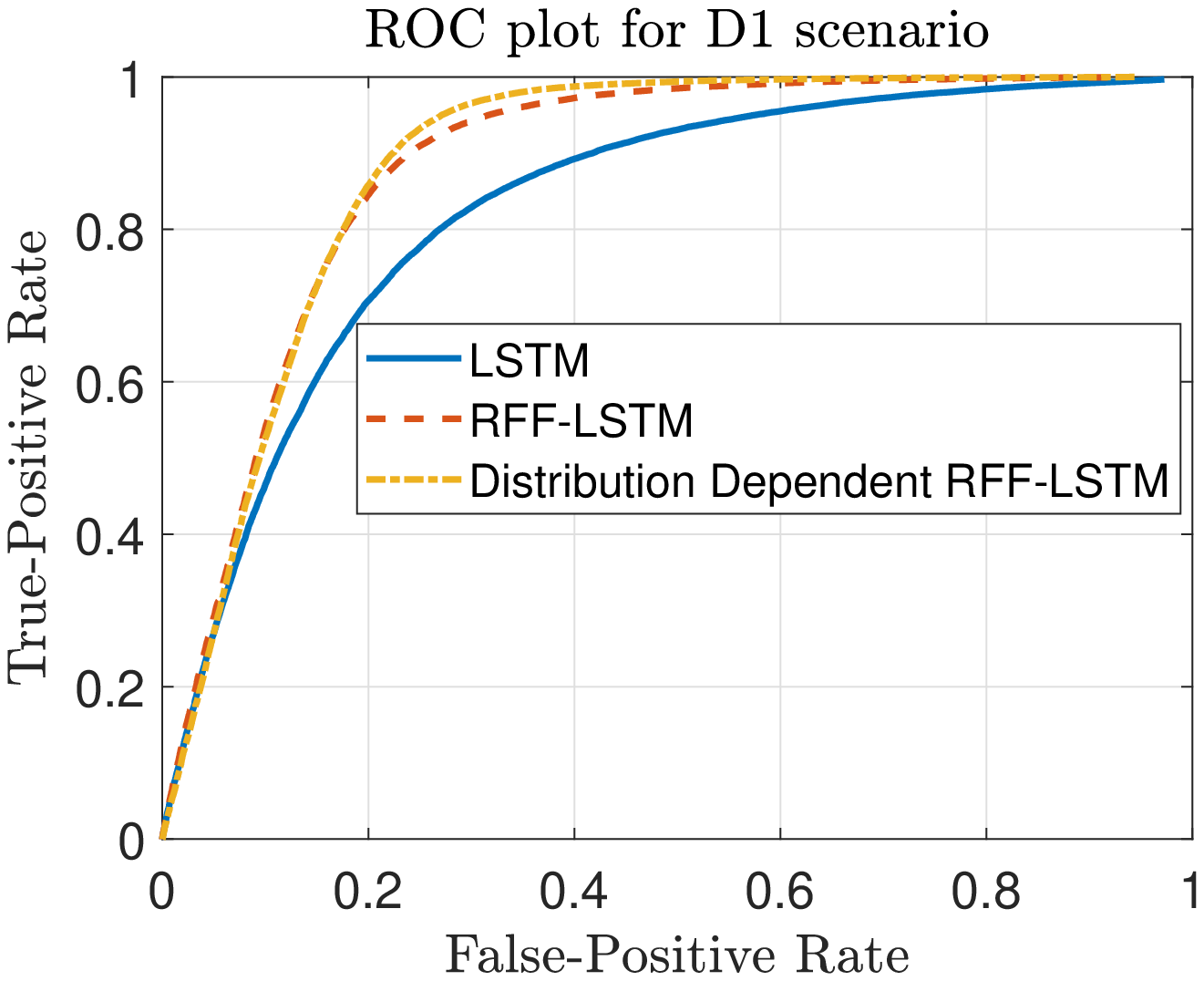}\\
  \caption{ROC curves obtained in the LOS/NLOS classification for LSTM, RFF-LSTM, and the Distribution dependent RFF-LSTM over suburban scenario D1 with a training-size of 400 observations.}\label{fig51}
\end{figure}
\begin{figure}[!htbp]
\centering
  \includegraphics[width=1\linewidth,height=7cm]{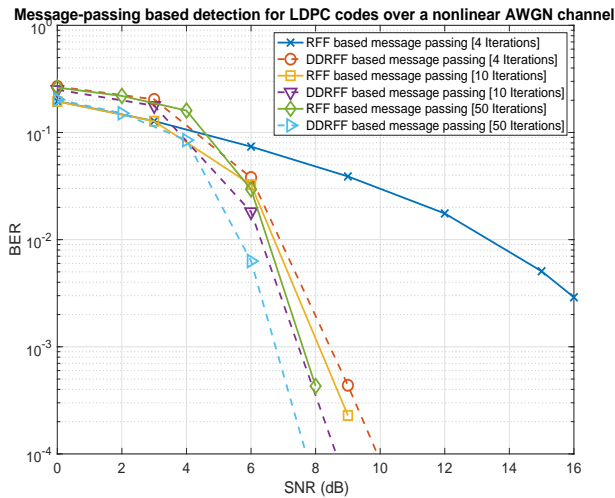}\\
  \caption{BER of LDPC decoding with message-passing over nonlinear VLC channels with: a) RFF and b) distribution-dependent RFF.}\label{fig6}
\end{figure}

A typical LOS VLC scenario is considered as in \cite{jain2019klms}, with Rapp LED nonlinearity, where the memory-parameter of the nonlinearity, $\alpha$, is considered to be 0.2, and the saturation current of the LED is 0.4. The generator matrix for the LDPC code is derived following the IEEE 802.11n standard with a block length of 648 \cite{910579,1255474}. From the BER results presented in Fig. \ref{fig6}, it can be inferred that the proposed distribution dependent RFF based message-passing outperforms the classical RFF based message-passing in terms of BER performance. Notably, the gains in BER performance are higher at lower number of outer iterations; though it is noted that even upon increasing the number of iterations to a very high value (like 50), we still get a significant gain with the proposed distribution-dependent RFFs, as seen in Fig. \ref{fig6}.
\section{Conclusion}
In this work, analytical results motivating the paradigm of RFF based DL are presented and a novel distribution-dependent RFF is proposed. The validity of the presented analysis and the proposed distribution-dependent RFF are ratified through realistic computer-simulations for critical machine-learning problems encountered in next-generation wireless communications, such as LOS/NLOS identification and LDPC decoding over nonlinear VLC channels. Simulations performed over realistic WINNER II outdoor channels validate the analytical proofs and indicate that the proposed neural network architecture delivers significant gains over classical LSTMs for LOS/NLOS identification, which makes the proposed methodology viable. Lastly, the worth of traditional RFF and the proposed distribution-dependent RFF maps are compared for message-passing based LDPC detection. In line with the derived theoretical results, the simulations indicate a significant performance gain upon deployment of the proposed distribution-dependent RFF, which enforces the usefulness of distribution-dependent RFFs for machine-learning applications for generic communication systems.
\section*{Acknowledgement}
This work was supported by the ULTRA TCS research chair on intelligent tactical wireless networks for challenging environments, and by the grant number CRDPJ 538896-19 from the National Natural Sciences and Engineering Research Council of Canada (NSERC).
\bibliographystyle{ieeetr}
\bibliography{./paper}
\end{document}